\documentclass[authoryear,11pt]{elsarticle}

\usepackage[colorlinks=true,linkcolor=black, citecolor=blue, urlcolor=blue]{hyperref}
\usepackage{multirow}
\usepackage{rotating}
\usepackage{tabularx}
\usepackage{graphicx}
\usepackage[normalem]{ulem}
\useunder{\uline}{\ul}{}
\usepackage[utf8]{inputenc}
\usepackage{booktabs, caption, makecell,dcolumn}
\usepackage[flushleft]{threeparttable}
\usepackage{adjustbox}
\usepackage{amscd}
\usepackage{amsmath,amsthm}
\usepackage{amsfonts}
\usepackage{amssymb}
\usepackage[capposition=top]{floatrow}
\usepackage{setspace}
\usepackage[top=1in, bottom=1in, left=1in, right=1in]{geometry}
\usepackage{footmisc}
\usepackage{todonotes}
\usepackage[]{algorithm2e}
\usepackage{xcolor}
\usepackage{verbatim}
\usepackage{comment}
\usepackage{nicefrac} 

\usepackage{appendix}
\usepackage{mathtools}
\DeclarePairedDelimiter\ceil{\lceil}{\rceil}

\newcommand{\ts}{\textsuperscript}
\newcommand{\RR}{\mathbb{R}}
\newcommand{\PP}{\mathbb{P}}
\newcommand{\cI}{\mathcal{I}}

\theoremstyle{plain}
\newtheorem{theorem}{Theorem}[section]
\newtheorem{definition}[theorem]{Definition}
\newtheorem{proposition}[theorem]{Proposition}

\theoremstyle{remark}
\newtheorem{remark}[theorem]{Remark}

\makeatletter
    \def\ps@pprintTitle{%
      \let\@oddhead\@empty
      \let\@evenhead\@empty
      \let\@oddfoot\@empty
      \let\@evenfoot\@oddfoot
    }
\makeatother

\makeatletter
\renewcommand{\paragraph}{\@startsection{paragraph}{4}{0ex}%
    {-3.25ex plus -1ex minus -0.2ex}%
    {1.5ex plus 0.2ex}%
    {\normalfont\normalsize\itshape}}
\makeatother
 
\stepcounter{secnumdepth}
\stepcounter{tocdepth}

\begin{document}

\title{Applications of Signature Methods to Market Anomaly Detection}

\begin{abstract}

\end{abstract}


\author[add1,add2]{Erdinc Akyildirim}
\ead{erdinc.akyildirim@math.ethz.ch}

\author[add1]{Matteo Gambara\corref{cor1}}
\ead{matteo.gambara@math.ethz.ch}


\author[add1]{Josef Teichmann\corref{cor1}}
\ead{josef.teichmann@math.ethz.ch}

\author[add1]{Syang Zhou}
\ead{syang.zhou@math.ethz.ch}

\address[add1]{Department of Mathematics, ETH, Zurich, Switzerland}

\address[add2]{Department of Banking and Finance, University of Zurich, Zurich, Switzerland}


\begin{keyword}
Machine Learning \sep Signature \sep Randomized Signature \sep Market Fraud \sep Reservoir Computing \sep Stock market \sep Cryptocurrency \sep Imbalanced data \sep Anomaly detection \\
{\it JEL:} C21 \sep C22 \sep G11 \sep G14 \sep G17 
\end{keyword}

\begin{abstract}
    Anomaly detection is the process of identifying abnormal instances or events in data sets which deviate from the norm significantly. 
    In this study, we propose a signatures based machine learning algorithm to detect rare or unexpected items in a given data set of time series type. We present applications of signature or randomized signature as feature extractors for anomaly detection algorithms; additionally we provide an easy, representation theoretic justification for the construction of randomized signatures. 
    Our first application is based on synthetic data and aims at distinguishing between real and fake trajectories of stock prices, which are indistinguishable by visual inspection. 
    We also show a real life application by using transaction data from the cryptocurrency market. In this case, we are able to identify pump and dump attempts organized on social networks with $F_1$ scores up to 88\% by means of our unsupervised learning algorithm, thus achieving results that are close to the state-of-the-art in the field based on supervised learning\footnote{The full code used for the anomaly detection case study on cryptocurrency data can be found under \href{https://gitlab.ethz.ch/Syang/reservoir-anomaly-detection}{https://gitlab.ethz.ch/Syang/reservoir-anomaly-detection}.}.\\
\end{abstract}

\maketitle

\newpage 
 
\section{Introduction}

Anomaly detection (or outlier detection) is defined as the task of detecting abnormal instances, which are the rare items, events or observations deviating significantly from the majority of the data. While these instances are called outliers (anomalies), the normal instances are called inliers. Anomaly detection is a fundamental research problem that has been investigated by researchers from diverse research fields and application areas. Anomaly detection can be made manually by searching through whole data clouds to diagnose the problem, but clearly this is a long and labour-intensive process. Anomaly detection often appears in the context of uncertainty, i.e.~absence, principal or not, of knowledge on the data generating process.

Hence, over time, a plethora of anomaly detection techniques ranging from simple statistical techniques to complex machine learning algorithms has been developed for certain application areas such as fraud detection in financial transactions (\cite{west2016intelligent}), fault detection in production (\cite{miljkovic2011fault}), intrusion detection in a computer network (\cite{sabahi2008intrusion}), etc. Some of the well known statistical methods such as z-score, Tukey method (Interquartile Range) or Gaussian Mixture models can be useful for the initial screening of outliers. Although these statistical or econometric anomaly detection methods have been well rooted in the literature (we refer the reader to \cite{chandola2009anomaly} for an extensive review) dating back to \cite{edgeworth1887xli}, many of them have failed to provide sufficient performance and accuracy in the last decade. This is mainly in view of big data collected from various sources such as financial transactions, health records, and surveillance logs etc. Nowadays high-volume, high-velocity, and high-variety data sets demand cost-effective novel data analytics for decision-making and to infer useful insights\footnote{ \href{https://www.gartner.com/en/information-technology/glossary/big-data}{https://www.gartner.com/en/information-technology/glossary/big-data}}.

Due to their inherent scaling properties machine learning algorithms take up the challenge of big data and seem to improve the success of anomaly detectors. Machine learning models can be configured to create complex anomaly detection systems with benefits such as higher performance, time saving, overall system stability, and cost saving in the long run. 
Machine learning algorithms for anomaly detection are essentially classified under two major paradigms: supervised or unsupervised, which depend on data at hand.

Under the supervised paradigm, data arrive labelled as anomaly and normal, that is all outliers are known in advance. Some of the frequently used machine learning algorithms under this setting are support vector machine (\cite{suykens1999least}), $k$-nearest neighbors (\cite{fix1989discriminatory}), Bayesian networks (\cite{pearl1985bayesian}), and decision trees (\cite{blfjors1984classification}). 

Although the supervised setting is ideal for the efficient application of ML algorithms, it is not common to obtain this structured data in most of real life applications. 
Under the unsupervised paradigm, in turn, the algorithm is not provided with any pre-labelled points in the training data hence the algorithm itself must discover what is anomalous and what is normal. Some of the frequently used machine learning algorithms in this case are neural networks (\cite{ivakhnenko1967cybernetics}), self-organizing maps (\cite{kohonen1982self}), expectation-maximization meta-algorithm (\cite{dempster1977maximum}), one-class support vector machine (\cite{choi2009least}). In our study we employ Minimum Covariance Determinant (MCD) (\cite{Rousseeuw1985}) and isolation forests (\cite{TonyLiu2008}), which are also unsupervised learning algorithms.

As mentioned above, the anomaly detection under the unsupervised paradigm (of unlabelled data) is a challenging task as it requires the algorithm to detect previously unseen rare objects or events without any prior knowledge about these. 
In this study, we consider time series data and we employ the powerful signature theory from rough path theory \cite{Lyons2014} as well as its randomized counterpart studied in \cite{sigSAS} together with isolation forests and robust covariance techniques to detect anomalies under the unsupervised paradigm.

Signatures provide a versatile encoding of the information of a continuous time bounded variation path. This can be used to construct a regression basis for non-linear, continuous maps on path space. We consider classification maps (in a relaxed way) as continuous maps on such path spaces. We demonstrate the application of (randomized) signature methodology for market anomaly detection in the financial markets, however, the same line of thought can also be applied in the other fields. 

Let us reflect a moment what an anomaly could mean in case of market data: as it is well known, the efficient market hypothesis (EMH) (\cite{fama1970}, \cite{sharpe1970efficient}, \cite{fama1976efficient}) depends on the random walk theory which suggests that price changes are random of a martingale type and thus cannot be predicted with certainty given all the currently available information, and the available information is optimally used by market participants. In other words, EMH states that stock prices (and their past) already reflect all available information and expectations about the value of the firm hence investors cannot earn excess profits using any past, public or private information which they obtain (legally). EMH has been one of the most widely researched and tested hypothesis in the last decades. It has been found that there are many instances where a security or group of securities performs contrary to efficient markets and these instances are viewed as anomalies. Financial market anomalies can occur in different extremes ranging from fraud or insider trading to calendar anomalies such as weekend, turn of the month, or January effects. There are also other type of technical and fundamental anomalies in between such as low price to book, high dividend yield, moving averages or momentum. 

The methodology developed in this paper allows to detect many different kinds of financial market anomalies mentioned above. 
In particular, we show the application of (randomized) signature methodologies for market anomaly detection by using simulated data (geometric Brownian motions with artificial changes) and also real transaction data from the cryptocurrency markets. 

One of the closest papers to our research in the mathematical framework is \cite{cochrane2020anomaly} which provides a data-driven notion of a distance (called conformance) between an arbitrary stream of data and the corpus. 
They combine conformance with signatures to analyse anomalies in streamed data. They achieve area
under the curve (AUC) score of more than 98\% (89\%) in the PenDigits (marine vessel traffic) data sets. 
They also show that their approach outperforms a state-of-the-art shapelet method for 19 out of 28 data sets for the univariate (non-financial) time series from the UEA \& UCR time series repository. 
Results from our methodology based on randomized signatures are mainly in line with their outcomes which shows the efficiency of the signatures based methodologies. 
On the other hand, from the finance literature, \cite{LaMorgia2021} is a very closely related work to our research. They investigate pump and dump and crowd pump type of manipulations organized by communities from the social media. They follow these groups for more than three years and collect around 900 manipulation cases. They show that their machine learning model is able to detect a pump and dump in 25 seconds from the moment it starts with an $F_1$-score of 94.5\%. We also employ this unique dataset of confirmed pump and dumps released by \cite{LaMorgia2021} to the literature. 
Although we achieve $F_1$ scores up to 88\%, the main advantage of our approach is that it is an unsupervised machine learning algorithm which does not require labelled data set, hence it can be applied in many different frameworks without data constraints.

\medskip

Our contribution to the literature is manifold: first, to the best of our knowledge, this is the first paper which employs randomized signatures and truncated signatures for anomaly detection in a financial market context. Second, we provide a new proof that qualifies randomized signatures as universal approximators in the space of continuous functions (under specific conditions). This complements the contribution of \cite{sigSAS}, which is more inspired by Reservoir Computing. Third, we contribute to the scarce literature on the detection of pump and dump schemes in the exponentially growing cryptomarkets. Moreover, our results reveal the improving market efficiency in the crypto ecosystem as the market manipulations can be detected with high accuracy rates in advance with our unsupervised machine learning algorithm soley relying on publically available data. Finally, the small differences between the results obtained by truncated signatures and randomized signatures empirically prove the robustness of the randomized signature method which can be applied in higher dimensional problems efficiently.

\medskip

The rest of the paper is organized as follows: Section \ref{sec:methods} explains the methodologies that we employ to detect the market anomalies. Section \ref{sec:example on synthetic data} presents the numerical applications using simulated and cryptocurrency market data together with the empirical findings and, eventually, Section \ref{sec:conclusion} concludes. 

\section{Methodology}\label{sec:methods}

The information contained in a time series can be encoded in non-linear maps on path space. Having several values of such a map enables its learning. In order to design learning algorithms efficient (linear or non-linear) regression bases have to be constructed. Signatures provide one universal solution for this problem, randomized signatures an efficient approximatively universal solution (with regularization properties).

\subsection{Signatures}

In this subsection, we provide some basics about signatures or signature transforms as well as the foundations of the theory of randomized signature. While a gentle introduction is also given by \cite{primerSignatures}, we also refer the reader to  \cite{Lyons2007} and \cite{friz_victoir_2010} for a detailed discussion on the signatures and rough paths.

In this section we shall introduce different ways to represent the information of paths, i.e.~functions $\gamma:I \to V$, where $I$ is a compact interval in $\RR$ and $V$ is a Hilbert space with an associated norm stemming from a scalar product $\|\cdot\|$. For simplicity we assume $V$ to be finite dimensional here, so tensor products coincide with their algebraic counterparts and can be used without ambiguity.

The \emph{total variation} of a path $\gamma$ is defined as $\|\gamma\|_{TV} = \sup_{\cI} \sum_{(t_1, \dots, t_k)\in \cI} \|\gamma(t_i) - \gamma(t_{i-1})\|$, where the supremum is taken over all partitions\footnote{That is, increasing sequences of ordered (time) indices such that $\cI = \{(t_0, \dots, t_r) \,|\, 0 = t_0 < t_1 < \dots < t_r = T \}$.} of $I$, called $\cI$.

\noindent Note that our path shall always start at $0$.

\begin{definition}[Bounded variation]
    A continuous path $\gamma: I \to V$ is said to be of \emph{bounded variation} if $\|\gamma\|_{TV} < +\infty$. 
    The set of all bounded variation paths on $V$ is denoted by $\operatorname{BV}(V) = \{\gamma:I\to V \,|\, \|\gamma\|_{TV} < +\infty\}$.
\end{definition}

Analogously, we can define the
\emph{p-variation} of a path:

\begin{definition}[$p$-variation]
    Let $p\geq 1$ be a real number and $\gamma:I \to V$ be a continuous path.
    The \emph{p-variation} of $\gamma$ on the interval $I$ is defined as 
    \begin{equation*}
        \|\gamma\|_p = \left( \sup_\cI \sum_{(t_1, \dots, t_k)\in \cI} \|\gamma(t_i) - \gamma(t_{i-1})\|^p \right)^{1/p}.
    \end{equation*}
\end{definition}

\noindent Note that for a fixed continuous path $\gamma$ the function $p \mapsto \|\gamma\|_p$ is non-increasing, hence, if $q \geq p$, then $\|\gamma\|_q \leq \|\gamma\|_p$ and, thus, any finite $p$-variation path has also finite $q$-variation.

\bigskip

One of the most remarkable properties of signatures is being universal approximators. To clarify this from two perspectives let us first define tensor algebras. For any integer $n$ the $n$-th power tensor of $V$ is defined as $V^{\otimes n} = V \otimes \dots \otimes V$. 
For example, for $V=\RR^d$, it holds that $V^{\otimes 2} = \RR^d \otimes \RR^d \simeq \RR^{d \times d}$.
By convention, $V^{\otimes 0} = \RR$.
\begin{definition}[Tensor algebra]
    Consider $V$ as a finite dimensional Hilbert space. 
    The \emph{extended tensor algebra} $T((V))$ over $V$ is defined as the space (notice the sum notation which we often apply)
    \begin{equation}
        T((V)) = \left\{ v = (v_0, v_1, \dots) = \sum_{n=0}^\infty v_n \,|\, v_n \in V^{\otimes n} \right\}.
    \end{equation}
    It is equipped with element-wise addition, with element-wise scalar multiplication (in the usual way) and with an inner product: $v\otimes w = (z_0,z_1, \dots)$ such that $V^{\otimes j} \ni z_j = \sum_{k=0}^j v_k \otimes w_{j-k}$ for all $j \geq 0$.
\end{definition}

We are now ready for the definition of a signature of a bounded variation path. 

\begin{definition}[Signature]\label{def: signature}
    Let $I=[s,t]$ be a compact interval and $X: I \to V$ be a continuous path with bounded variation starting at $0$. Consider furthermore a basis $ e_1,\ldots,e_d$ in $V$ and denote by $X^i$ the coordinate of $X$ in this basis.

    Let $\mathbf{i} = (i_1, i_2, \dots, i_n)$ be a multi-index of length $n$ where $i_j \in \{1,\dots, d\}$, for all $j \in \{1,2,\dots, n\}$. 
    Define the coordinate signature of the path $X$ associated to the multi-index $\mathbf{i}$ as iterated integral as follows
    \begin{equation}\label{eq: sig_term}
        S(X)^\mathbf{i} e_{i_1} \otimes \dots \otimes e_{i_n} = \int_{s \leq u_1 \leq \dots \leq u_n \leq t} dX_{u_1}^{i_1} \dots dX_{u_n}^{i_n} e_{i_1} \otimes \dots \otimes e_{i_n} \, .
    \end{equation}
    Then the \emph{signature} $S(X)$ of $X$ is defined as
    \begin{equation}\label{eq: sig_tuple}
        S(X) = (1, S(X)^{(1)}, \dots, S(X)^{(k)}, \dots),
    \end{equation}
    where $S(X)^{(k)} = \sum_{\mathbf{i}=(i_1, \dots, i_k)} S(X)^\mathbf{i} e_{i_1} \otimes \dots \otimes e_{i_n} = \int_{s \leq u_1 \leq \dots \leq u_n \leq t} dX_{u_1} \otimes \dots \otimes dX_{u_k}$ for any $k\geq 1$ has been defined in \eqref{eq: sig_term}.
    In a similar way, we introduce the \emph{truncated signature} of $X$ of degree $N$:
    \begin{equation}
        S_N(X) = (1, S(X)^{(1)}, \dots, S(X)^{(N)}).
    \end{equation}
\end{definition}

\begin{remark}
Signatures can be defined beyond bounded variation paths. For bounded variation its definition is possible by classical Lebesgue-Stiltjes integration theory. For $p$-variation $1 < p <2$, signature can be defined by Young integration, see \cite{Lyons2007}. For path of $p$-variation with $ p > 2 $ one needs -- depending on $p$ -- additional data from the path $X$. This data turn out to be components of signature up to $[p]$. We shall only work with signatures of bounded variation paths here, but we note that we could also define -- in complete analogy to the above and in line of signatures for geometric rough path stemming from a continuous semi-martingale -- signature of a continuous semi-martingale via Stratonovich integration.
\end{remark}

\begin{remark}\label{rem: number coef exact signature}
    For a path in $V=\RR^d$ there are $d^N$ iterated integrals of order $N$. 
    This implies that the size of signatures grows exponentially in degree.
    In particular, the truncated signature of degree $N$ has in total $\sum_{j=0}^N d^j = \frac{d^{N+1}-1}{d-1}$ for $d \geq 1$.
\end{remark}

Analytically, signature of a path is just defined as the collection of all iterated integrals of components of the path with itself. Signature also appears as a dynamical system driven by the path $X$ on the interval $I=[0,t]$ for running $t$, and has therefore several geometric properties. We can also easily interpret components of signature: the first levels are simply the increments of the path, while the second levels are linked to L\'evy areas drawn by paths (see \cite{primerSignatures}).

An important property of signatures is Chen's identity (\cite{learningFromPast2016}):
\begin{theorem}[Chen's identity]
    Let $X$ and $Y$ be paths with bounded variation defined on consecutive time intervals and starting at $0$, e.g.~$[0,t]$ and $[t,T]$, for $0\leq t \leq T$.
    Then it holds that
    \begin{equation}
        S(X \star Y) = S(X) \otimes S(Y),
    \end{equation}
    where the concatenation of the two paths $X$ and $Y$ is defined as
    \begin{equation*}
        (X \star Y)_u = 
            \begin{cases}
                X_u, &\mbox{ for } u \in [0,t],\\
                X_t + Y_u, &\mbox{ for } u \in [t,T].
            \end{cases}
    \end{equation*}
\end{theorem}
\noindent This fundamental relation stemming from the fact that signature is related to a canonical dynamical system in $T((V))$ is the key ingredient to calculate (truncated) signatures with the computer as  it is implemented in the Python package \textsc{iisignature} (see \cite{iisignature}).
In fact, it turns out that for piecewise linear paths, we have
\begin{equation*}
    S(X_1 \star \dots \star X_m) = \bigotimes_{j=1}^m S(X_j) = \bigotimes_{j=1}^{m} \Big ( \sum_{n=0}^\infty \sum_{\mathbf{i}=(i_1,\ldots,i_n)} \frac{1}{n!} X^{i_1}_j \cdots X^{i_n}_j e_{i_1} \otimes \dots e_{i_n} \Big)  \quad ,
\end{equation*}
i.e.~it can be calculated from standard polynomials of components of the linear paths $X_j$ by taking products in $T((V))$.
\bigskip

A priori, it is not clear whether truncated signatures are able to perform well enough for our purposes, as they simply contain only a fraction of the information of the ``entire'' signature.
Nonetheless, as the following theorem shows, this is completely analogous to polynomial expansions (see Chapter 2 of \cite{Lyons2007}).
\begin{theorem}
    Let $X:I=[s,t] \to V$ be a path of bounded variation.
    Then, for each $n \geq 0$, one has
    \begin{equation}
        \left|\int_{s \leq u_1 \leq \dots \leq u_n \leq t} dX_{u_1} \otimes \dots \otimes dX_{u_n}\right| \leq \frac{\|X\|_{TV}^n}{n!}.
    \end{equation}
\end{theorem}

\noindent Thus, high-order signature terms contain less and less information, since we have a factorial decay, and truncated signatures should express similar features as standard signatures, provided the truncation degree is sufficiently high.
\noindent The link between normal signatures and truncated signatures is given by the so-called \emph{canonical projection}:
\begin{definition}[Canonical Projection]
    The canonical projection $\pi_N$ of an element of $T((V))$ on the truncated tensor space $T^N(V)$ is defined by $T((V)) \to T^N(V)$, $(v_0, v_1, \dots, v_N, v_{N+1}, \dots) \mapsto (v_0, v_1, \dots, v_N)$.
\end{definition}

\noindent Exploiting this projection, we see that the truncated signature solves an ordinary differential equation (ODE) in finite dimensions, i.e.~it is a dynamical system, that can be directly derived from signature definition (we refer to Chapter 7 of \cite{friz_victoir_2010} for a proof and more details).

\begin{proposition}\label{prop: truncated_sig_ode}
    Let us fix the interval $I=[s,t]$ and let $X:I \to V$ be a path of bounded variation and $Y:I \to T^N(V)$. Then the solution of the controlled ODE 
    \begin{equation}
        \begin{cases}
            dY_u = \pi_N(Y_u) \otimes dX_u \, , \text{ for } u \in [s,t] \\
            Y_s = (1, 0, 0, \dots)
        \end{cases},
    \end{equation}
    is the signature of $X$ truncated at order $N$. Appropriately interpreted one can choose $N=\infty$ in this assertion, too.
\end{proposition}

Since $T((V))$ is an associative, non-commutative algebra freely generated by $\operatorname{dim} V$ generators with unit $1$ signatures, as defined in \eqref{eq: sig_tuple}, there are elements in such algebra solving an ODE controlled by the path $X$.

From the point of feature extraction we view signature as a map $S: \operatorname{BV}(V) \to T((V))$.  The analytic properties of signatures allow for the following remarkable statement (which is just the Leibniz rule expressed in a fancy way): any product of two components of a signature is a linear combinations of components of signature: \emph{therefore the set of linear combinations of signature components forms an algebra of functions on $\operatorname{BV}(V)$}. Algebraically speaking $T((V))$ is additionally a Hopf algebra and signature is a group-like element therein.

Additionally, if we shall always include running time as the zeroth component of a path $X$, then the map from $ X $ to its signature is actually injective. 
This can be seen in the following way: knowing the signature of a path means to know in particular
$$
\int_{s \leq u_0 \leq \ldots \leq u_{n-1} \leq v \leq t} dX^0_{u_0} \ldots dX^0_{u_{n-1}} d X^i_{v} = \frac{1}{n!} \int_s^t v^n d X^i_v 
$$
for some fixed interval $ [s,t]$, for $ n \geq 0$, and $ i \in \{1, \ldots, d \}$, whence 
all Fourier coefficients of the path are determined and therefore the path itself (notice again that the path starts at $0$!).

Whence linear combinations of components of signature form a point separating subalgebra of continuous functions on $\operatorname{BV}(V)$ and by the Stone-Weierstrass theorem, for any compact set $K$ of continuous paths of bounded variation, the set of linear functionals on signatures of paths from $K$ is dense in the set of continuous real-valued functions on $K$, denoted by $C(K, \RR)$:

\begin{theorem}
    For any function $f \in C(K, \RR)$, $K\subset \operatorname{BV}(V)$ compact and $\varepsilon > 0$, there exists a $l \in T(V)$, the dual space of $T((V))$ (the free vector space of words in $\{ 0,\ldots,d \}$ or the space of non-commutative polynomials in $e_0,\ldots,e_d$), such that $\sup_{h \in K} \| f(h) - \langle l, S(h)\rangle \|< \varepsilon$.
\end{theorem}

\noindent By using the above theorem signatures are a complete set of feature extractors and are therefore universal approximators (\cite{Kiraly2019}).

However, in any application these infinite dimensional objects have to be approximated by finite dimensional objects. 
Therefore, in practice, we will never deal with the full signatures but use a truncated version of those, which are defined in \emph{truncated signatures} (introduced in Definition \ref{def: signature}). An alternative approach to this standard procedure is presented in the sequel.

\subsection{Randomized Signatures}

Instead of the abstract tensor algebra $T((V))$ we can consider concrete realizations, or -- algebraically speaking -- representations.

Consider $V$ being generated by $e_0,\ldots,e_d$ ($e_0$ will correspond to time), then we can consider maps from $e_i$ to vector fields on some $\mathbb{R}^k$ 
$$
e_i \mapsto \sigma(A_i \,\cdot\, + b_i)
$$
for $ k \times k $ matrices $ A_i$ and vectors $ b_i \in \mathbb{R}^k $ and some real analytic, bounded activation function $\sigma$ which is applied componentwise.

\begin{definition}
Let $X: [0, T] \mapsto \mathbb{R}^{d+1}$ be a continuous path of bounded variation augmented by time in the $0$-th component. Then the solution of
\begin{equation}\label{eq: ode randomized signature}
    \begin{cases}
        d\operatorname{RS}_t = \sum_{i=0}^d \sigma(A_i \, \operatorname{RS}_t + b_i) \,dX^i_t \\
        \operatorname{RS}_0 \in \mathbb{R}^k
    \end{cases},
\end{equation}
is called \emph{randomized signature}.
\end{definition}
If the above representation from the free algebra generated by $ e_0,\ldots, e_d $ to the algebra of differential operators on $\mathbb{R}^k$ is faithful, i.e.~injective, then randomized signature for any initial value should contain precisely the same information as the signature itself:
\begin{theorem}
Let  $\sigma$ be real analytic with infinite radius of convergence and let $A_1,\ldots,A_d$, $b_1,\ldots,b_d$ be independent samples of a probability law absolutely continuous with respect to Lebesgue measure. 
Then linear combinations of randomized signature for all possible initial values $\operatorname{RS}_0 \in \mathbb{R}^k$ are dense in $C(K)$ on compact subsets of bounded variation curves.
\end{theorem}

\begin{proof}

By construction the vector fields $\sigma(A_i\,\cdot+b_i)$ cannot satisfy any non-trivial relation with probability one, since otherwise we would have constructed a non-zero real analytic function on some $ \mathbb{R}^{(d+1) \times k^2 + (d+1) \times k}$, the product space of $d$ matrices and $d$ vectors with product measure equivalent to Lebesgue measure, whose zero set has positive Lebesgue measure. This is impossible by \cite{mityagin2015zero}. Therefore the representation of $T((V))$ into real analytic vector fields on $\mathbb{R}^k$ is faithful. Compare here very similar arguments in \cite{cuclartei:20} in a slightly different but comparable context of generalized H\"ormander or Chow theorems.

By boundedness of $ \sigma$ randomized signatures are well defined and by the fact that $\sigma$ is real analytic with infinite radius of convergence we obtain by \cite{Bau:12} that the stochastic Taylor expansion converges and represents $\operatorname{RS}_t$ for any $t > 0$. 
Whence signature components can be expressed by randomized signature components for different initial values and we are done.

\end{proof}

\begin{remark}
In practise we choose $k$ moderately high, but we usually do only calculate $\operatorname{RS}$ for one initial value. One could improve in view of the above proof the approximation properties by calculating $\operatorname{RS}$ for different initial values.
\end{remark}

\begin{remark}
Notice also that we could take any free set of $d+1$ vector fields on any geometric structure and construct randomized signatures thereon. Notice also that different initial values of $ \operatorname{RS} $ just corresponds to different shifts $ b_i $. Stacking together these equations leads to one initial value for a moderately high dimensional system of type \eqref{eq: ode randomized signature}.
\end{remark}

The use of exact, truncated or randomized signatures as a feature set is very closely related to the concept of reservoir computing which models input - output systems as

\begin{equation}\label{eq: reservoir computing}
    \begin{cases}
        x_t = F(x_{t-1}, z_t) \\
        y_t = h(x_t)
    \end{cases},
\end{equation}
where $F$ is some mapping of inputs $z_t$ and previous features $x_{t-1}$ to the next feature $x_t$. $h$, usually called \emph{readout} or \emph{observation} map in the context of reservoir computing, fulfills the role of mapping the features, in our case either truncated exact signatures or randomized signatures, to the output.
The activation function $\sigma$ in  \eqref{eq: ode randomized signature} plays the role of a \emph{squashing function} which ensures that the results do not blow up.
A detailed discussion on reservoir computing and state-affine systems together with some important theoretical results can be found in \cite{echostate} and \cite{sigSAS}. 
In line with this interpretation, we will call $k$ the reservoir dimension.

\bigskip

Next, we will present, how randomized and exact (truncated) signatures can be used in practice for classification tasks. 
For this, we will first present an example based on simulated data and afterwards a case study based on real data, in which we aim to identify market frauds using signatures. 
In the remainder of this article, we will write signature for either exact or randomized signatures and we will specify it explicitly when necessary. 

\section{Numerical Applications}
\label{sec:example on synthetic data}
In this section, we show the applications of our methodology first by using the simulated price paths from geometric Brownian motion and later with some real data obtained for cryptocurrency pump and dumps cases.

\subsection{Simulated Data}
\label{sec: synthetic data}
In this section we illustrate how our machine learning algorithm with signatures can be used to differentiate the price paths that follow a geometric Brownian motion from the price paths coming from a fake Brownian motions which contain artificial forgeries.

Our synthetic data set consists of sample paths from a geometric Brownian motion  (GBM in the following) of the form 
\begin{equation}\label{eq: gbm}
    dX_t =\mu X_t dt + \sigma X_t dW_t,
\end{equation}
and artificially manipulated geometric Brownian motions. 
In general, manipulated Brownian motions do not exhibit all typical behaviours from real Brownian motions. 
Particularly, in successful market manipulations, some market participants have knowledge about the price paths, which is not publicly available. 
This manifests itself, for instance, in a lack of randomness for the periods of market manipulation.

\begin{figure}
    \centering
    \includegraphics[width=\textwidth,height=\textheight,keepaspectratio]{"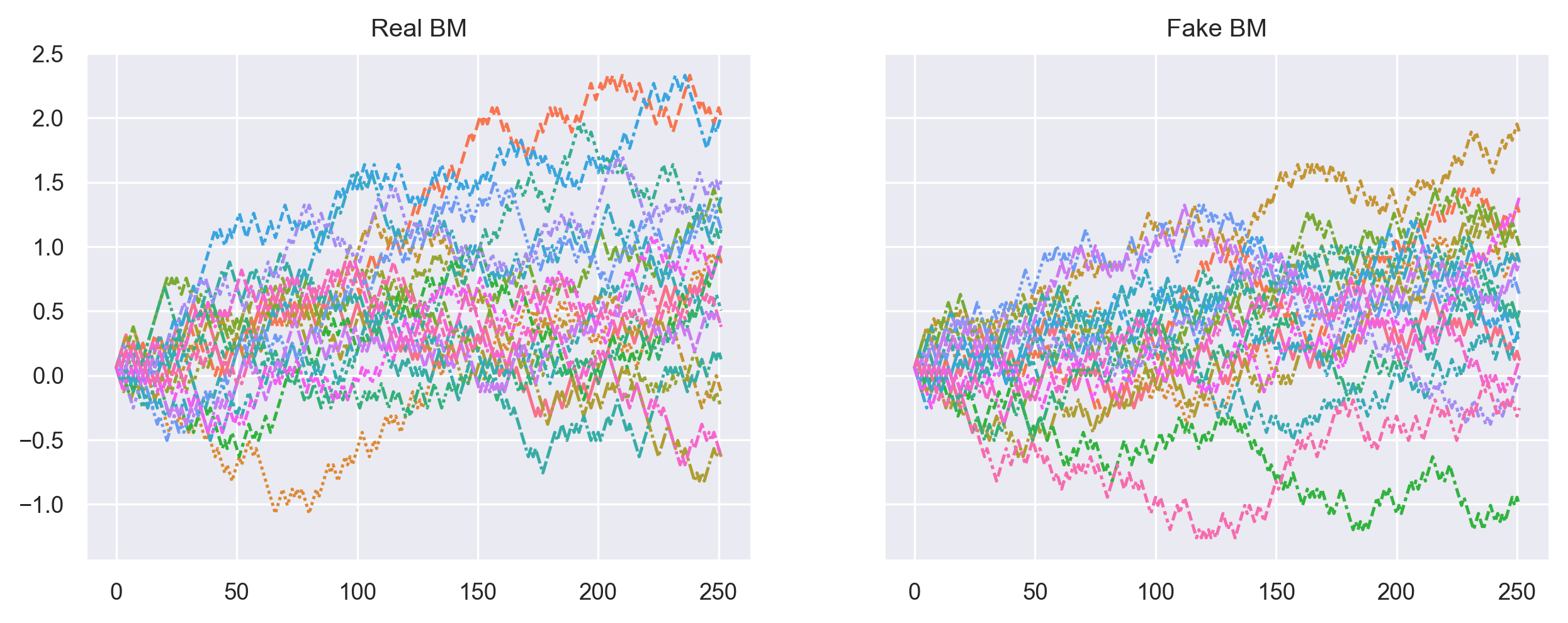"}
    \caption{Sample paths of Geometric Brownian Motion (GBM) used as synthetic data. Parameters are
    $\mu=0.25, \sigma=0.2$ and each path, representing 1 year, is made by 252 steps. All trajectories are shifted subtracting $S_0=100$.}
    \label{fig: sample paths generated}
\end{figure}

Using this fact, we choose to generate artificially manipulated Brownian motions in the following way:
Given price differences $R_i = X_i - X_{i-1}$, a pattern $P$ of the form $P_n = (p_1, \ldots, p_n), \quad p_i \in \{-1, 1\}$ and a partition $\mathcal{P}$ with subintervals of length $n$, we prevent the occurrence of $P$ in the return path. 
In order to ensure that the expected mean value of the paths do not change, we always also prevent the occurrence of $-P_n = (-p_1, \ldots, -p_n)$ in the return path. 
More concretely, our synthetic data consists of 20'000 paths, equally divided in real and manipulated, which are then split into train and test sets. 
Longer patterns $P$ will make the problem more difficult, as there will be less differences between paths to distinguish them. 
In our experiments, we choose to remove patterns of length 6 or higher.

\begin{figure}[H]
    \centering
    \includegraphics[trim={0 0 0 0},width=0.9\textwidth]{"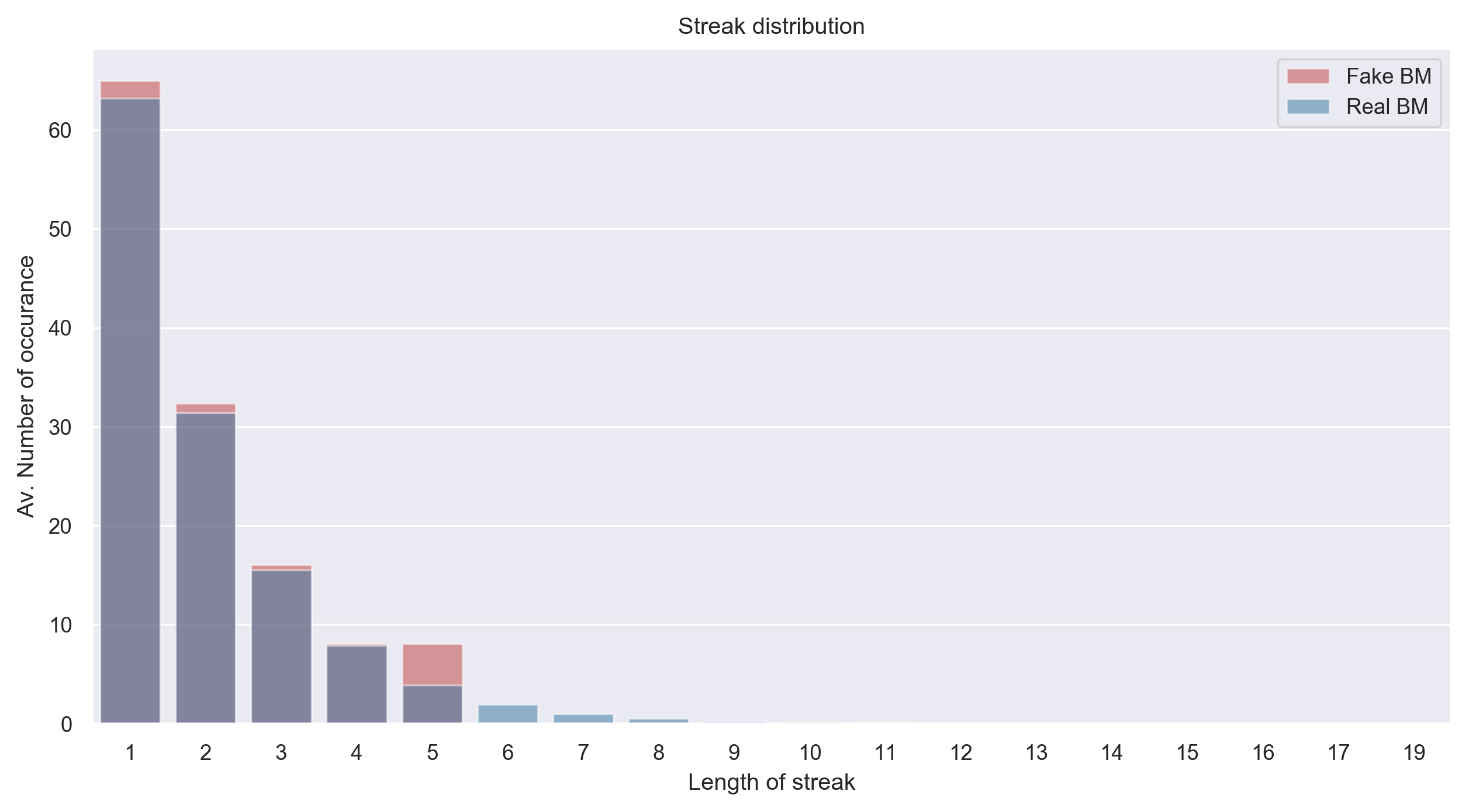"}
    \caption{Histogram of the number of consecutive increases/decreases in the paths. Fake GBM do not have any combination of more than 5 consecutive increases/decreases along their paths.}
    \label{fig: streaks illustration}
\end{figure}

Samples paths of the fake geometric Brownian motions (GBM) and real Brownian motions can be found in Figure \ref{fig: sample paths generated}. From this figure it can be seen that it is not easy to understand the difference between real and fake Brownian motions by visual inspection. Also obvious characteristics of the paths like mean and variance will be preserved and therefore they cannot be easily used to differentiate the paths.

The manipulated Brownian motions can of course be differentiated by counting the number of long streaks as illustrated in Figure \ref{fig: streaks illustration}. As the number of occurrences of streaks does not differ substantially between  real  and  fake  Brownian  motions,  it  is  not  a trivial  task to  differentiate  them.

\subsection{Feature set generation}
We follow the methodology developed in Section \ref{sec:methods} to transform the simulated data described into features which are then used by a logistic regression. We add the time dimension and the first differences to the path yielding $(t, X_t, \Delta X_t)$ as the 3-dimensional time series which we use for the classification task. 
The first differences are added as our experiments have shown that they improved robustness and performance. 
Note that the information content does not change with the addition of the first differences.
Table \ref{tab: hyper paramaters simulated} gives a full overview for all parameters used to generate the results in Section \ref{sec:results simulated}. 
The features are finally generated by solving \eqref{eq: ode randomized signature} with the hyper-parameters described in Table \ref{tab: hyper paramaters simulated} using a forward Euler scheme. 
In this example, we compute the random signatures for each path and use a logistic regression on the training set. 
The weights are then used to evaluate the performance of the models. 
For each path, model outputs will be in [0,1] which indicates the probabilities. 
The model labels each path as either a real geometric Brownian motion or a fake geometric Brownian motion by rounding the model outputs.

\vfill

\begin{table}[ht]
\caption{Hyper parameters for random signatures}
\centering
\begin{tabular}{llll}
\toprule
Symbol&Description & Simulated data & Cryptocurrency data \\
\midrule
$\sigma$& Activation function in \eqref{eq: ode randomized signature}&  $\tanh$ &   $\tanh$ \\
$\mu_A$ & Mean of $A$ in \eqref{eq: ode randomized signature}            &  0.15 &  0.05\\
$\sigma^2_A$ & Variance of $A$ in \eqref{eq: ode randomized signature} &    0.6   & 0.1 \\
$\mu_b$  & Mean of $b$ in \eqref{eq: ode randomized signature}          &   0 &  0 \\
$\sigma^2_b$ & Variance of $b$ in \eqref{eq: ode randomized signature}&       1 &    1  \\
$w$ &Window size as described in \ref{sec: feature set generation crypto}    &      - &       100  \\
$o$ &Offset length as described in \ref{sec: feature set generation crypto}  &       - &       5  \\
$k$ & Reservoir dimension& 200 &  50       \\
\bottomrule
\end{tabular}
  \label{tab: hyper paramaters simulated}
  \newline

       {\raggedright Table \ref{tab: hyper paramaters simulated} shows the results all hyper parameters used for both the simulated data and the cryptocurrency data for the generation of random signatures.\par}
\end{table}

\vfill

\subsubsection{Results}
\label{sec:results simulated}
\begin{figure}[H]
    \centering
      \caption{ROC curve (left) and PR curve (right) for the classification task of differentiating paths of the form \ref{fig: sample paths generated}. Both figures show very similar results for the exact truncated signatures and the random signatures.}
\includegraphics[width=\textwidth,keepaspectratio]{"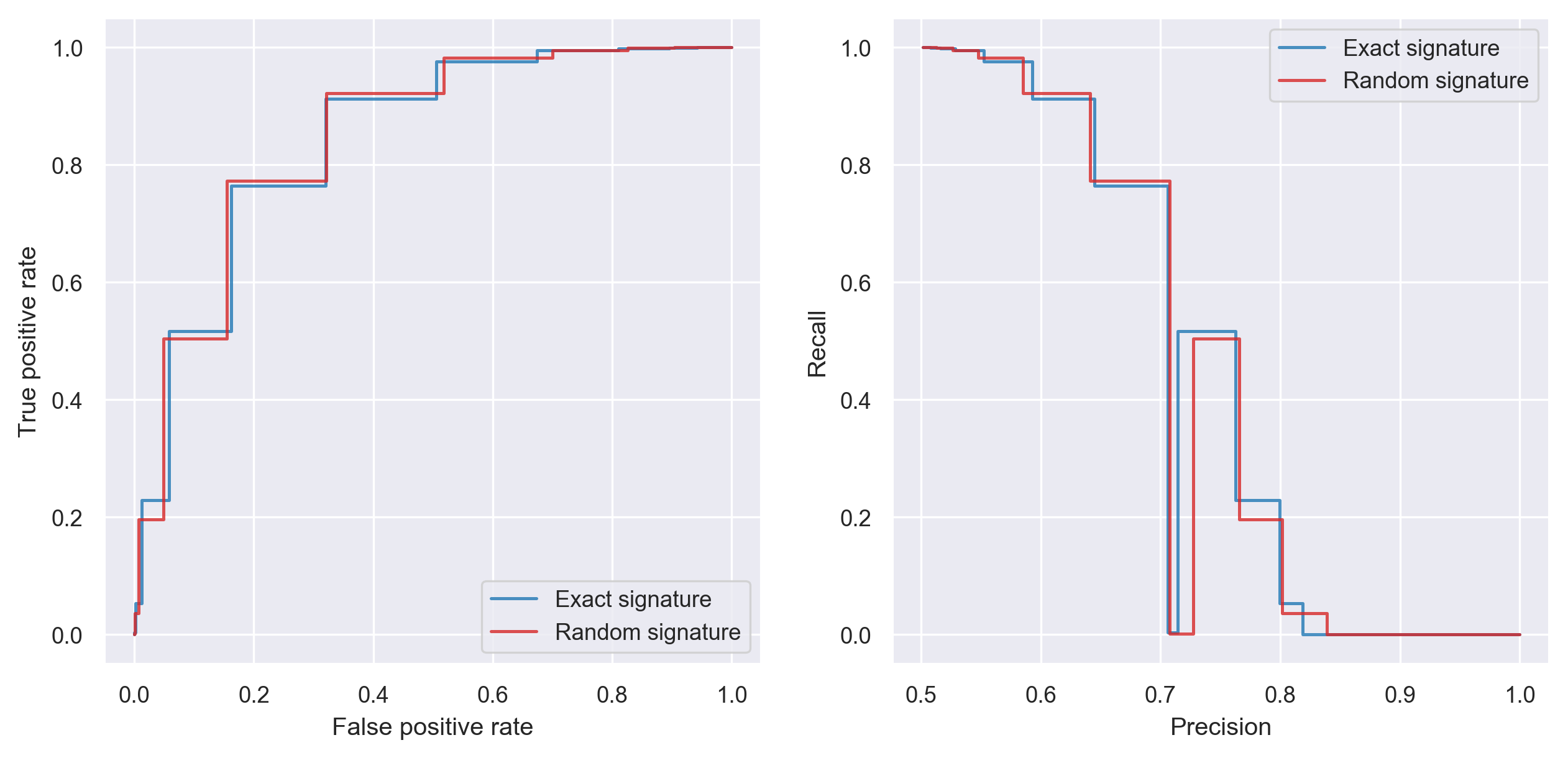"}
    \label{fig: balanced roc pr plot}
\end{figure}

We evaluate the performance of our methods by comparing the total accuracy. 
For that, it is useful to introduce the following definitions that are common in classification (and pattern recognition):
\begin{definition}[Precision]
    Precision is the number of true positives (i.e.~the number of items correctly labelled as belonging to the positive class) divided by the total number of elements labelled as belonging to the positive class (i.e.~the sum of true positives and false positives, which are items incorrectly labelled as belonging to the class).
    Briefly, $p = \nicefrac{tp}{(tp+fp)}$.
\end{definition}
\begin{definition}[Recall]
    Recall in this context is defined as the number of true positives divided by the total number of elements that actually belong to the positive class (i.e. the sum of true positives and false negatives, which are items which were not labelled as belonging to the positive class but should have been).
    Briefly, $r = \nicefrac{tp}{(tp+fn)}$.
\end{definition}
\begin{definition}[Accuracy]
    Accuracy is the proportion of correct predictions (both true positives and true negatives) among the total number of cases examined.
    Briefly, $a = \nicefrac{(tp+tn)}{(tp+tn+fp+fn)}$.
\end{definition}
In our experiments, model performances obtained on the train set were similar to the test set, which indicates good generalisation properties. 
Additionally, we also compared the accuracy scores for different regression outputs. This can be seen in Figure \ref{fig: window plot}. 
This figure confirms that regression outputs closer to 0 and 1 yield higher accuracy scores. 
Notably the results for the top quantiles are better than for the bottom quantiles. 
This can be explained by the fact that in our particular way of modifying Brownian motions, the existence of long streaks will be a very strong indicator that the path is a real Brownian Motion whereas the absence of long streaks does not always imply that the path has been manipulated as strongly. 
Therefore, our methods show better results in correctly identifying real Brownian motions than in detecting manipulated Brownian motions.

Next, we compare the model performance quantitatively. The results are summarised in Table \ref{tab: comparison simulated data}. 
Comparing the methods, we observe a similar outcome for randomized signatures and exact truncated signatures. 
\begin{definition}[ROC curve]
    ROC (Receiver Operating Characteristics) graphs are two-dimensional graphs in which \emph{tp}-rate (defined as the recall) is plotted on the y axis and \emph{fp}-rate (defined as $\nicefrac{fp}{(fp+tn)}$) is plotted on the x axis. An ROC graph depicts relative tradeoffs between benefits (true positives) and costs (false positives).
\end{definition}
\begin{definition}[PR curve]
    PR (Precision Recall) graphs are two-dimensional graphs in which \emph{recall} is plotted on the y axis and \emph{precision} is plotted on the x axis. The focus of the PR curve on the minority class makes it an effective diagnostic for imbalanced\footnote{As it will be the case for cryptocurrency data in Section \ref{sec: crypto}} binary classification models.
\end{definition}

\noindent Also the ROC curve and the PR curve both show very similar results for randomized signatures compared to exact signatures as seen in Figure \ref{fig: balanced roc pr plot}. 
This indicates the validity of our approach as the results suggest that information content in the exact truncated signatures can be retained by using randomized signatures. The results shown serve to demonstrate, that both random signatures and exact signatures with even a simple readout function can be successfully used on this sample problem.

\begin{figure}[H]
    \centering
    \includegraphics[width=\textwidth,keepaspectratio]{"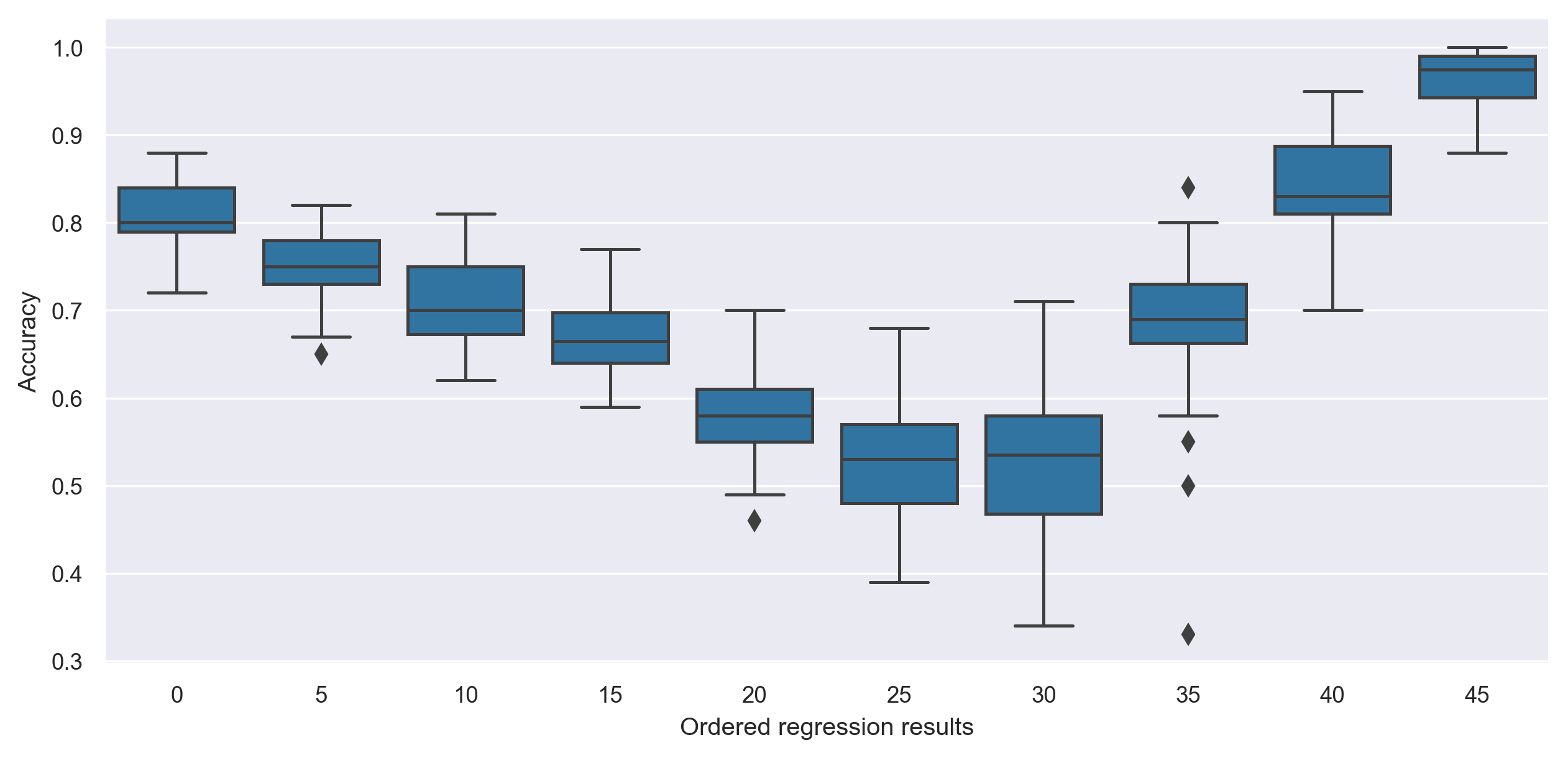"}
    \caption{The x-axis shows ordered model outputs for 10 separate quantiles. The y-axis shows the accuracy of given set of model outputs. Here one can see that accuracy is higher for model outputs, which are farther from the cutoff.}
    \label{fig: window plot}
\end{figure}

\begin{table}[ht]
\caption{Real signature vs random signature}
\centering
\begin{tabular}{lll}
\toprule
{} & Exact signature & Randomized signature \\
\midrule
Mean accuracy            &        71.81\% &         72.4\% \\
Mean accuracy bottom 10\% &       83.56 \% &       84.95 \% \\
Mean accuracy top 10\%    &       95.11 \% &       95.74 \% \\
\bottomrule
\end{tabular}
  \label{tab: comparison simulated data}
  \newline

       {\raggedright Table \ref{tab:  comparison simulated data} shows the results for the randomized and the truncated exact signatures for the sample problem. Here top/bottom $10\%$ displays the accuracies when only taking the top/bottom $10\%$ of the regression results respectively.\par}
\end{table}

\subsection{Cryptocurrency pump and dumps}\label{sec: crypto}
Due to the unregulated nature of cryptocurrency brokers, cryptocurrencies have been the target of all sorts of market manipulation attempts. 
In this case study, we will focus on pump and dump schemes.
Typically, pump and dump schemes in the cryptocurrency space are organised via private chat groups in channels using different softwares, such as Telegram or Discord, where any interested person can join the channels and receive the date and currency pair of the planned pump and dump schemes. 
The exact currency pair is given only seconds before the start of the pump, typically via OCR (optical character reading-resistant) resistant images similar to captchas.
There are even aggregate web pages, which collect information of pump and dump channels and display them on their web page, see for example Figure \ref{fig: pumpolymp} taken from \href{https://pumpolymp.com/}{pump olymp}. 
We would like to emphasize that we do not recommend using such websites to participate in pump and dump schemes. As can be seen in the following process description, there is great uncertainty with this strategy, with participants having the least amount of information. 
This makes following such strategies extremely risky besides the moral considerations.
\begin{figure}[h!]
  \caption{Example of a web page collecting pump and dump channels}
  \label{fig: pumpolymp}
  \includegraphics[width=\textwidth]{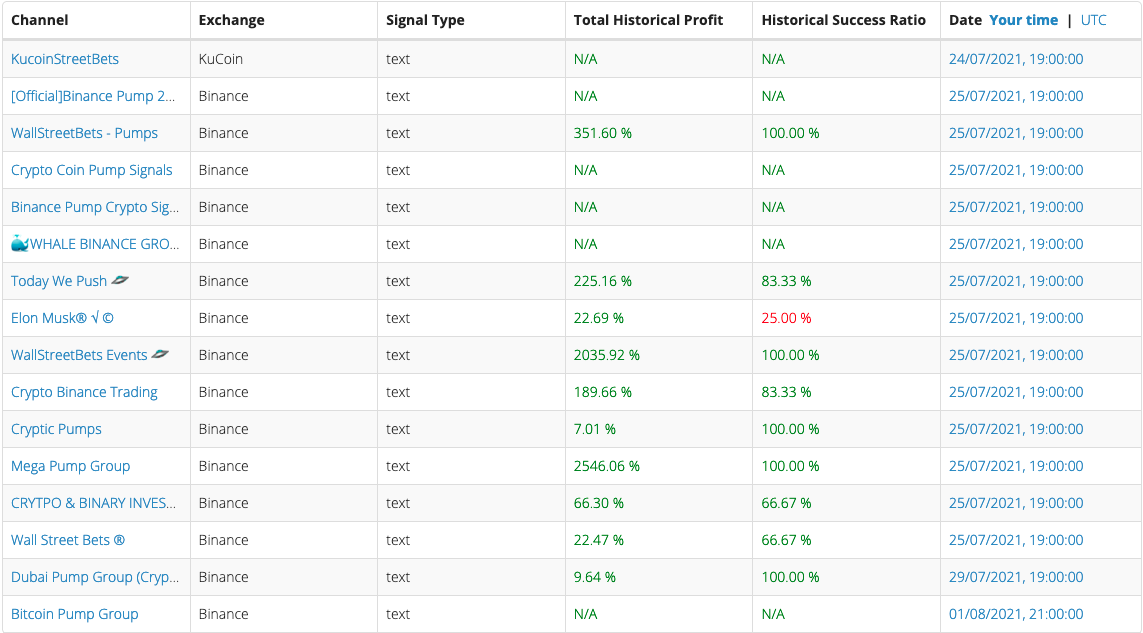}
\end{figure}

In the following we will describe the main ingredients and process of a typical cryptocurrency pump and dump scheme. A more detailed description of the phenomenon can be found in \cite{Xu2019} and \cite{LaMorgia2021}.
\subsection{Pump and Dump ingredients}

\begin{itemize}
    \item Organisers: Most often, pump and dump organisers are admins of the corresponding
    Telegram / Discord channels and use various means to advertise upcoming pumps. 
    They are the ones who also determine the exchange and the currency pair which will be the target of the pump and dump scheme.
    \item Participants: These are the traders who participate in pump and dump schemes in order to generate abnormal returns by using non-public information.
    They are subscribed to the channels and receive the information about the exchange and the currency pair targeted. 
    It is worth noting that participants are usually categorised into different tiers depending on their financial contribution to the organisers or their help in spreading the channel; in other words, there is an internal hierarchy.
    Participants in different tiers receive the information about the exchange and currency pair at different time points, where participants in the lowest tier receive the information last. 
    This obviously creates an extremely unbalanced situation where participants in higher tiers have a way higher chance to generate profits.
    \item Exchange: Cryptocurrencies are typically traded in crypto exchanges, which function as brokers. 
    Due to the currently unregulated nature of these exchanges, they do not have too much incentive to stop pump and dump schemes as they generate profits for them due to transaction costs. 
    As a matter of fact, some exchanges such as Yobit has have openly organised pump and dump multiple times \cite{Xu2019}.
    \item Currency pair: The currency pair is a price quote of the exchange rate for two different cryptocurrencies (e.g. BCT/ETH is the price of Bitcratic denominated in Ethereum). 
    Pump and dump schemes usually target coins with a low market capitalisation so that the price can be influenced significantly with relatively small trades.
\end{itemize}

\subsection{Pump and Dump process}\label{sec: pd process}

Typically a pump and dump scheme has three phases, which will be described in the following.
\begin{itemize}
    \item Pre-pump activity: Organisers are advertising the upcoming pump via their respective channels. 
    On predetermined times, see Figure \ref{fig: pumpolymp}, a countdown is started in order to reveal the currency pair being targeted. 
    In order to mask automatic detection by bots, the information is typically revealed as manually modified pictures, which are difficult to read for machines.
    \item Pump: The organisers urge their members to buy the coin and hold it in order to inflate the price on the short term. 
    If the pump is successful, this leads to an immediate price spike, which can often also be recognised by human eye.
    \item Dump: After the coin has drastically increased in price, invariably some participants will start consolidating their gains by selling the coin. 
    This leads to the price decreasing which then leads to a cascading effect of more and more participants selling the coin.
    This often results in the price returning to the level close to the pre-pump level.
\end{itemize}

This whole process takes place in a matter of minutes if not seconds. 
In fact, among the pump and dump schemes we consider and which we could classify based on graphical comparison, the longest one lasts  8 minutes and the average duration is about 2.3 minutes. 
In this case study, we follow the methodology introduced by  \cite{LaMorgia2021} and \cite{Kamps2018}. \cite{LaMorgia2021}\footnote{We would like to personally thank the authors for the long and tedious work done and for the choice of making these data freely available at \href{https://github.com/SystemsLab-Sapienza/pump-and-dump-dataset}{https://github.com/SystemsLab-Sapienza/pump-and-dump-dataset}.} manually collect pump and dump cases by joining the various Telegram / Discord channels. We are also using these cases in order to verify the results from our algorithm.

\subsection{Data analysis}\label{sec: data analysis}

As mentioned above, we take our data labels from the database created by \cite{LaMorgia2021}, which they obtain by collecting information joining more than 100 groups from July 2017 to January 2019. 
The database consists of a list of what we call as Pump-and-Dump attempts, which we will abbreviate as PD attempts for the rest of this article. Every entry of the list is made by
\begin{enumerate}
    \item the acronym for the name of the coin, e.g.~BTC stands for Bitcoin and ETH for Ethereum, usually called \textbf{symbol} of the coin;
    \item the name of the Telegram / Discord \textbf{group} where this information was retrieved;
    \item the \textbf{date} (day) of the PD attempt;
    \item the \textbf{time} (hour and minute) of the PD attempt;
    \item the \textbf{exchange} on which the PD attempt took place.
\end{enumerate}
Using this information and the Python-library \href{https://ccxt.readthedocs.io/en/latest/index.html}{\textsc{ccxt}}, it is possible to fetch data from exchanges in a neighbourhood of the PD attempt. 
In this way, we avoid downloading an overwhelming quantity of data which would be even difficult to store. 
For every entry of the aforementioned database, we decided to download and work with trades, that is all orders that were placed and also realized by selling or buying cryptocurrencies.
Every trade is characterized by the following information:
\begin{enumerate}
    \item \textbf{Symbol} which denotes the traded coin. 
    In practice, this is not just one acronym, but it refers to the traiding pair, e.g.~BTC/USD, where Bitcoins (BTC) are traded using USD (US Dollars);
    \item \textbf{Timestamp}, an integer, denoting the UNIX time in milliseconds;
    \item \textbf{Datetime}, i.e.~the ISO8601 datetime with milliseconds;
    \item \textbf{Side}, string, either ``buy'' or ``sell'' to distinguish the operation kind: for the symbol BTC/USD ``buy'' means buying BTC using USD, while ``sell'' means receiving USD for your BTC;
    \item \textbf{Price}, float, indicating the price at which the pair was traded;
    \item \textbf{Amount}, float, denoting the quantity (in base currency, i.e.~USD in the example of BTC/USD pair) which was traded.
\end{enumerate}
For the moment, all analyses we conducted are made using pairs of the type $\cdot\,$/BTC, where BTC is seen as the base currency.
For this reason, why we will not explicitly write the entire pair in the following. 
We decided to download all trades in different time intervals: by considering the moment of the PD attempt as the center of a time window, we retrieved data for intervals $\pm$ 1.5, $\pm$ 3 days and $\pm$ 7 days, obtaining for every PD attempt intervals of 3, 6 and 14 days respectively.
Increasing here the analysis for $\pm$ 7 days has to be treated carefully, as there might be manipulations which took place without being recorded in the database\footnote{In particular, as the authors of \cite{LaMorgia2021} write, they were not able to retrieve information from groups in Russian or Chinese.}. 
We also observed that this is indeed the case for some coins (see, for example, Figure \ref{fig: sample2 pd detection}).

The numerical data we used for the analyses required some preprocessing on the raw data. 
First of all, we aggregated the data in the same way as in \cite{LaMorgia2021}, in other words whenever trades have the same timestamp, side and price, then the trade volumes are summed accumulated. In that case, one can interpret the trades as one larger trade without losing any information.

Furthermore, the variable side was translated into a numerical format by encoding all ``buy'' signals as 0.5 and all ``sell'' as -0.5.
In a similar way, we calculate \textbf{volume} as the product of price and amount.
Moreover, prices are used to compute the simple returns.
Hence, we use the time series of the trade time, returns, volume and trade side information  to compute the signatures.
Finally, our experiments showed that the normalization of all of these variables to the [0,1] interval yield better results.

A visual example of the data that we use is given in Figure \ref{fig:GRScoin}. We plot the price (top panel) and volume (bottom panel) against prescribed time interval for the Pump \& Dump attempt, which was in this case successful.
In particular, we show the price and volume charts of GRS (Groestlcoin) against the base currency (BTC) and we plot all the trades which occur in the time interval 31 minutes before and 71 minutes after the prescribed PD time. On the figure, ``buy'' (``sell'') trades are shown in red (green).

\begin{figure}
    \caption{GRS price and volumes around Pump \& Dump (successful) attempt which was recorded on 5\ts{th} March 2019 at 16:30 (UNIX time) - a visual example. In the top plot the lines follow the price evolution through the different trades, while in the bottom plot we can see the volume in terms of the basis currency (BTC). In both cases, red colour is used for `buy' trades and the green for `sell' trades.}
  \label{fig:GRScoin}
  \centerline{\includegraphics[trim={0 0 0 0}, width=1.15\textwidth]{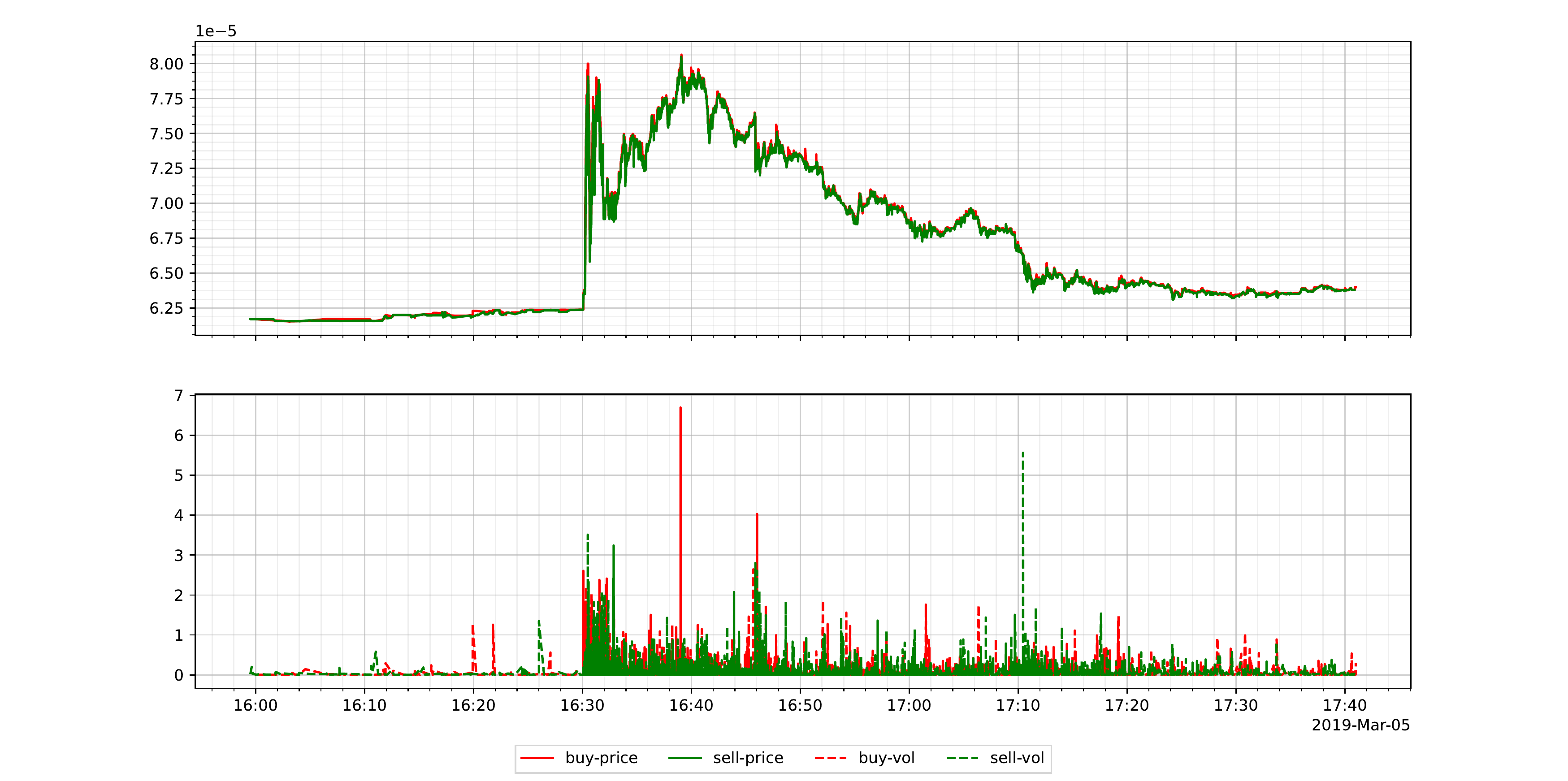}}
\end{figure}

\subsection{PD detection}
Our methodology for detecting the PD attempts has two main steps. 
As a first step, we use the data described in the Section \ref{sec: data analysis} and we transform it into our features set. 
In the next step, we use this features set together with the commonly used anomaly detection algorithms to produce our final predictions.

\subsubsection{Feature set generation}
\label{sec: feature set generation crypto}
We follow the methodology developed in Section \ref{sec:methods} to transform the time series data described in Section \ref{sec: data analysis} into the features which can be used by the standard anomaly detection algorithms.
Let 
\begin{equation*}
    D_c \subset \mathbb{R}^{n \times 4}
\end{equation*}
denote the vector of time series composed of trades' timestamp, side, price, and volume for a coin $c$. 
Here $n$ denotes the number of trades during the pre-specified time interval around the PD attempt which is chosen to be 1.5, 3 or 7 days.
We split the data set $D_c$ into $\ceil{n/o}$ subsets of $w$ trades, where $o$ denotes the offset of how many trades to move forward in time. 
In other words, we identify the trade-window specified by $w$ trades and then we move on shifting among the listed trades.
\begin{equation*}
    D_c = \bigcup_{i=1}^{\ceil{n/o}} D^i_c.
\end{equation*}

Each subset $D^i_c$ is then transformed into our feature set by calculating either the exact signature or the randomized signature. 
Let $\Tilde{R}^i_c$ and $R^i_c$ denote the randomized and exact signatures, respectively. 
This process is repeated for all the coins in the data set so that we construct our final feature set
\begin{align*}\label{eq: feature set}
    \Tilde{\mathbf{R}} &= \bigcup_c \bigcup_i^{\ceil{n_c/o}} \Tilde{R}^i_c, \\
    \mathbf{R} &= \bigcup_c \bigcup_i^{\ceil{n_c/o}} R^i_c,
\end{align*}
where the number of trades $n_c$ depends on the coin $c$.  
In this methodology, the offset $o$ and the window size $w$ can be seen as hyper-parameters which can be optimized to improve the final results. 

Note that especially the offset $o$ can have a big impact on the run time as the number of samples per coin will be $\ceil{n/o}$. (Our numerical experiments show that choosing $w=100$ and $o=5$ yields a good balance between the run time and precision.

Including the time dimension,  which is clearly non-decreasing, in the features set also helps to obtain better behaved signatures, in the sense that the signatures in this case are able to uniquely determine the original time series (see Lemma 2.14 and Lemma 4.8 in \cite{learningFromPast2016} for a proof\footnote{Essentially, we are avoiding tree-like paths (trajectories that cancel out by themselves as controls).}). The hyper-parameters used for the feature generation from cryptocurrency data can be found in table \ref{tab: hyper paramaters simulated}.

\begin{remark}
    Actually, since we are working with finite data on finite intervals, we can use Fubini's theorem to exchange the order of integration inside the signature terms.
    In particular, we can move the integral on time as if it were the last integration.
    This shows that every signature will clearly encode information generated by paths integrated over time every particular time interval.\\
    For this reason, it is important to note that the time intervals identified by the first and last trades of the selected window are not uniform throughout our analysis.
    In fact, because trades are not equidistant in time, each sample could potentially represent a different time length.
    Working with a heterogeneous grid of time is not easy a priori (and this is why \cite{LaMorgia2021} first transform the data into an equidistant time grid. 
    However, the non-equidistant time grid does not pose a problem to our algorithm, because by including the timestamp in the data set, the correct time scale is reflected in the calculation of the randomized and exact truncated signatures.
\end{remark}

\subsection{Anomaly detection algorithms}

Anomaly detection involves being able to distinguish whether an element of the sample belongs or not to the same distribution which generated the (greatest part of the) sample.
Elements of this latter group are usually called \emph{inliers}, while elements of the former are commonly referred to as \emph{outliers}.
The task is of course difficult because we cannot rely on visual inspection for multi-dimensional data and we do not have a priori a clear description of the ``exact'' distribution and its samples.
In the literature there are a number of anomaly detection algorithms, each with its own pros and cons,
but most commonly used algorithms share the following steps:
\begin{enumerate}
    \item Give a feature set $F\subset \mathbb{R}^{n\times m}$ consisting of $n$ samples and $m$ features, each sample is interpreted as a point in $\mathbb{R}^m$.
    \item Using some metric, for each such point the distance to neighboring points is calculated.
    \item Based on the density of points, a 1-dimensional anomaly score is calculated where either higher or lower values correspond to more abnormal observations depending on the exact algorithm.
    \item Lastly, based on a predetermined contamination rate, a cutoff is set and all samples with an anomaly score higher or lower than the cutoff are declared as anomalies.
\end{enumerate}
In our case study, we use the robust covariance and isolation forest anomaly detection algorithms.
These methods are implemented in the Sklearn Python package (\cite{scikit-learn}). 
We refer the reader to \cite{Rousseeu1999} and \cite{TonyLiu2008} for more details about the used algorithms.

\subsubsection{Robust covariance}
The idea behind the first one is quite simple: the goal is to compute elliptical contours inside which inscribing the inliers. 
The ellipsoid is constructed through a robust covariance matrix and has form:
\begin{equation*}
    E(\Sigma, \mu, \rho) = \left\{ x\in \RR^m \,:\, (x-\mu)'\Sigma^{-1}(x-\mu) \leq \rho^2 \right\}.
\end{equation*}
The robust covariance matrix is built using Minimum Covariance Determinant (MCD) estimator whose goal is finding which elements in the complete sample yield the empirical covariance matrix with the smallest determinant, suggesting a ``regular''  subset of observations from which the location and covariance matrix are preserved.
In particular, we need to fix a fraction $0<\gamma<1$ and consider a subsample of observations (hopefully, inliers), i.e.~$I \subset \{O_1,\dots, O_n\}$ that contains $\ceil{n\gamma}$ points.
Subsample mean and covariance are given by
\begin{align*}
    \widehat{\mu_I} &= \frac{1}{|I|}\sum_{O_i \in I} O_i,\\
    \widehat{\Sigma_I} &= \frac{1}{|I|}\sum_{O_i \in I} (O_i-\widehat{\mu_I})(O_i-\widehat{\mu_I})'.
\end{align*}
The ellipsoid determined by points in $I$ is $E(\widehat{\Sigma_I}, \widehat{\mu_I}, \widehat{\rho_I})$, where $\widehat{\rho_I} = \inf_{r>0} \PP_I(E(\widehat{\Sigma_I}, \widehat{\mu_I}, r)\geq \gamma)$, with $\PP_I$ being the empirical measure associated to $I$.
Finally, Mahalanobis distance with the robust covariance is used to rank outliers. 
The algorithm is looking for the covariance matrix with lowest determinant because the volume of the ellipsoid is proportional to the square root of the determinant.
The MCD algorithm also has a high breakdown point, which is equal to $[(n-m+1)/2]$.\\
Inconveniences of this method are the fact that data are supposed to belong to an elliptical distribution and that to compute the covariance matrix (in a sufficiently reliable way) we need the number of samples larger than the square of the number of features ($n>m^2$)\footnote{This condition is actually requested by the algorithm in the Sklearn package.}, which imposes another condition to the truncation degree used in our algorithm, but which can easily accounted for in case of randomized signatures.\\
The algorithm was mainly developed by Rousseeuw, starting from \cite{Rousseeuw1985} and it became popular in 1999, thanks to \cite{Rousseeu1999}, where a faster version of the same procedure was formulated. 
For extensions and other results, we suggest \cite{MCD_Extensions} (and references therein).

\subsubsection{Isolation forest}
The second anomaly detection algorithm we use is called isolation forests and was first published in \cite{TonyLiu2008}. 
The idea of this algorithm is to recursively partition the data until either all unique points are separated or a predefined limit is reached. 
The recursion of partitions is expressed as a \emph{proper binary tree}, where each node has either two or zero descendants. 
\begin{definition}
The path length $h(x)$ of a point $x$ is the height of the point $x$ is the random tree.
\end{definition}
Intuitively points, which are anomalies, are easier to separate from the other points, which results in lower height scores $h(x)$. 
The reason for this is that anomalies lie outside of clusters, which is why a random partition will end up isolating the point more quickly and hence yield a lower height score $h(x)$. 
An isolation forest consists of a sample of these random trees where the height scores are averaged. The height $h(x)$ can not directly be used as an anomaly score, because the average height is given by $\log n$ where $n$ is the number of points in the input data. As such, the height can not be directly compared across different samples and a normalisation is required.
Given an input data set of $n$ points, we define
\begin{equation*}
    c(n) = 2H(n-1) - (2(n-1)/n)
\end{equation*}
where $H(i)$ is the harmonic number given by $\sum_{k=1}^{i} \frac{1}{k}$. The anomaly scores is then defined as 
\begin{equation*}
    s(x, n) = 2^{-\frac{E(h(x))}{c(n)}}
\end{equation*}
where $E(h(x))$ is the averaged height from the isolation forest. As discussed, anomalies should yield lower $E(h(x))$, which is why high scores of $s(x, n)$ indicate anomalies. Interestingly, isolation forest work better with \emph{smaller} data sets, which is why sub-sampling methods are used.
Note that this is directly contrary to the robust covariance method, where, for fixed number of features, more
points are preferable to construct the ellipsoid. 
This also give an intuitive reason, why isolation forests work well with both truncated exact signatures and random signatures, as is discussed in the subsequent section.

\subsection{Results}
\begin{figure}
  \caption{Example of labeled anomalies by signature (green diamond) and benchmark predictions (green circle), while in purple it is time window of the PD attempt.}
  \label{fig: sample1 pd detection}
  \includegraphics[scale=0.7]{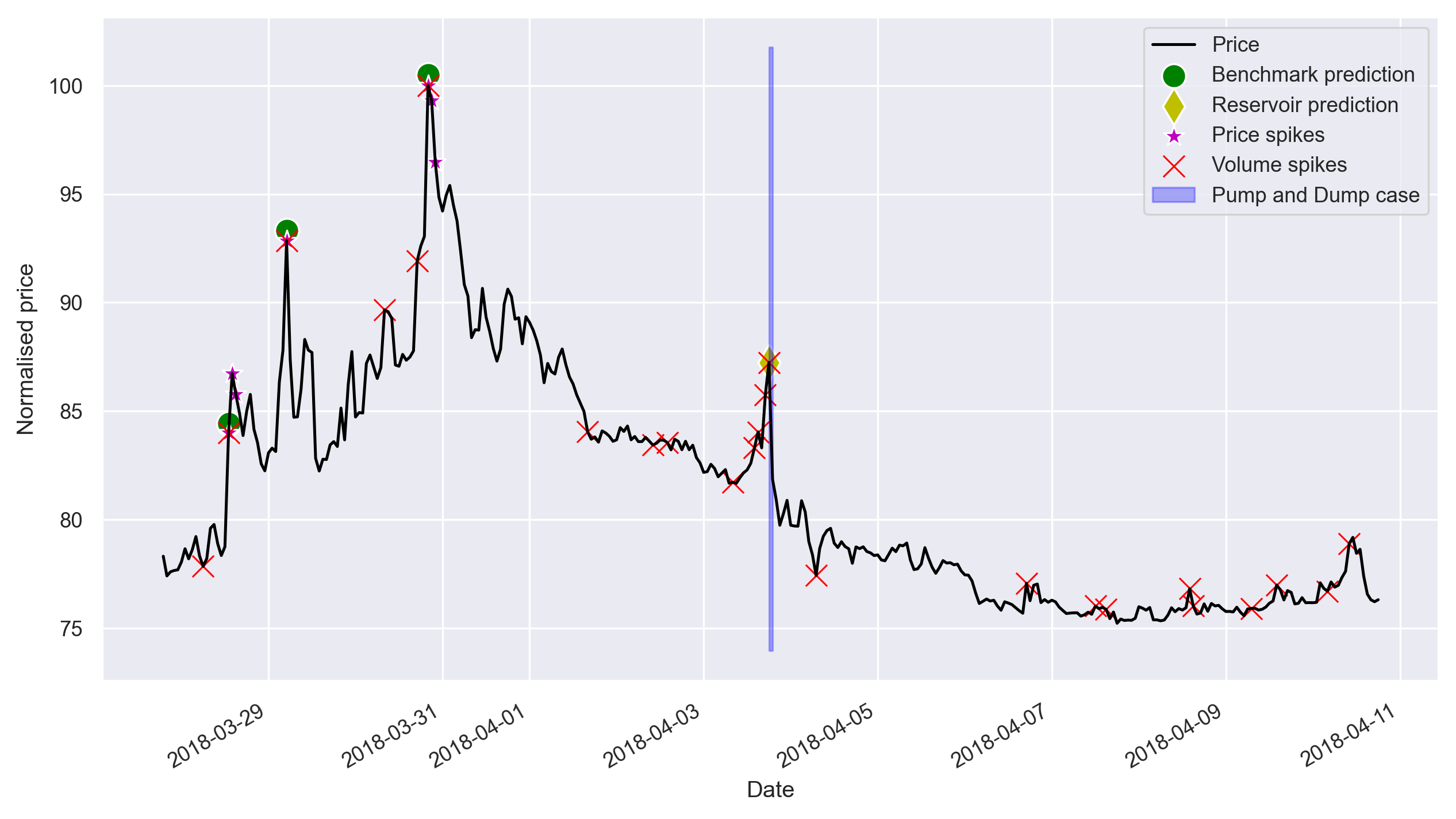}
\end{figure}

In order to evaluate the quality of our anomaly detection algorithm, we compare our predictions with the manually labeled PD attempts. 
We evaluate our predictions by considering 1.5, 3, and 7 days before and after each PD attempt. 
Choosing a longer time interval makes the predictions more useful but at the same time carries the risk of PD attempts to be absent from our data set.
As a benchmark, we use the anomaly detection method employed by \cite{Kamps2018}).
In this method, anomalies are labeled by considering volatility spikes and price spikes.
They describe the best configuration  setting as the volume spike threshold to be $V_{spike} = 300\%$ and the price spike threshold to be $S_{spike} = 105\%$. 
This means that time windows are labeled as anomalies, if they contain both a volume and price increase higher than the threshold compared to a moving average. In order to obtain a full precision-recall plot, we linearly interpolate the threshold $V_{spike}$ between $0\%$ and $300\%$ and $S_{spike}$ between $0\%$ and $105\%$. 
This also enables us to compare the algorithms on the full precision-recall plot while ensuring that the benchmark configuration is always considered. As the main metric to evaluate our results we consider the maximum of the $F_1$ scores across the precision-recall plot.

\begin{definition}[$F_1$ score]
    The $F_1$ score is defined as the harmonic mean between precision and recall, i.e.~$F_1 = \nicefrac{2}{(p^{-1} + r^{-1})}$. 
    It measures of the classifier's accuracy by merging precision and recall in one metric and its values range between 0 (worst case) and 1 (base case). 
    Since it is defined as a harmonic mean, it gives more weight to smaller ``values'', thus a classifier has a high $F_1$ score only when both precision and recall are high.
\end{definition}

Following \cite{Kamps2018}, we aggregate our results to hourly windows. Figure \ref{fig: sample1 pd detection} shows an example where the predictions using signatures produce more accurate results as there are less false positives compared to the benchmark. This result is not surprising as the benchmark model produces many false positives since it does not quantify by how much the volatility and price thresholds are breached. Additionally, the benchmark prediction was not able to detect the real PD attempt, as the price spike was not high enough to breach the threshold.

Figure \ref{fig: sample2 pd detection} shows a less successful example where, even though the reservoir prediction correctly labels PD attempts, it also incorrectly labels other time windows as PD attempts. 
Visual inspection of the price plot shows that both the reservoir prediction and the benchmark prediction incorrectly marks another significant price spike as a PD attempt. 
Additionally, there is also a wrong classification by the reservoir prediction around 21\textsuperscript{st} July 2018 in the figure where an unusual price drop occurs. 
This can be explained by the fact that the reservoir prediction is an unsupervised learning method hence there is no way of learning that unusual price dip should not be considered as an anomaly. 

\begin{figure}
  \caption{Example of mislabelling by signature (green diamond) predictions close to the line denoting 21\textsuperscript{st} July 2018, while in purple the time window of the PD attempt. The other outlier identified by signatures before 23\textsuperscript{rd} July might be a correct guess that was not recorded in the dataset.}
  \label{fig: sample2 pd detection}
  \includegraphics[scale=0.7]{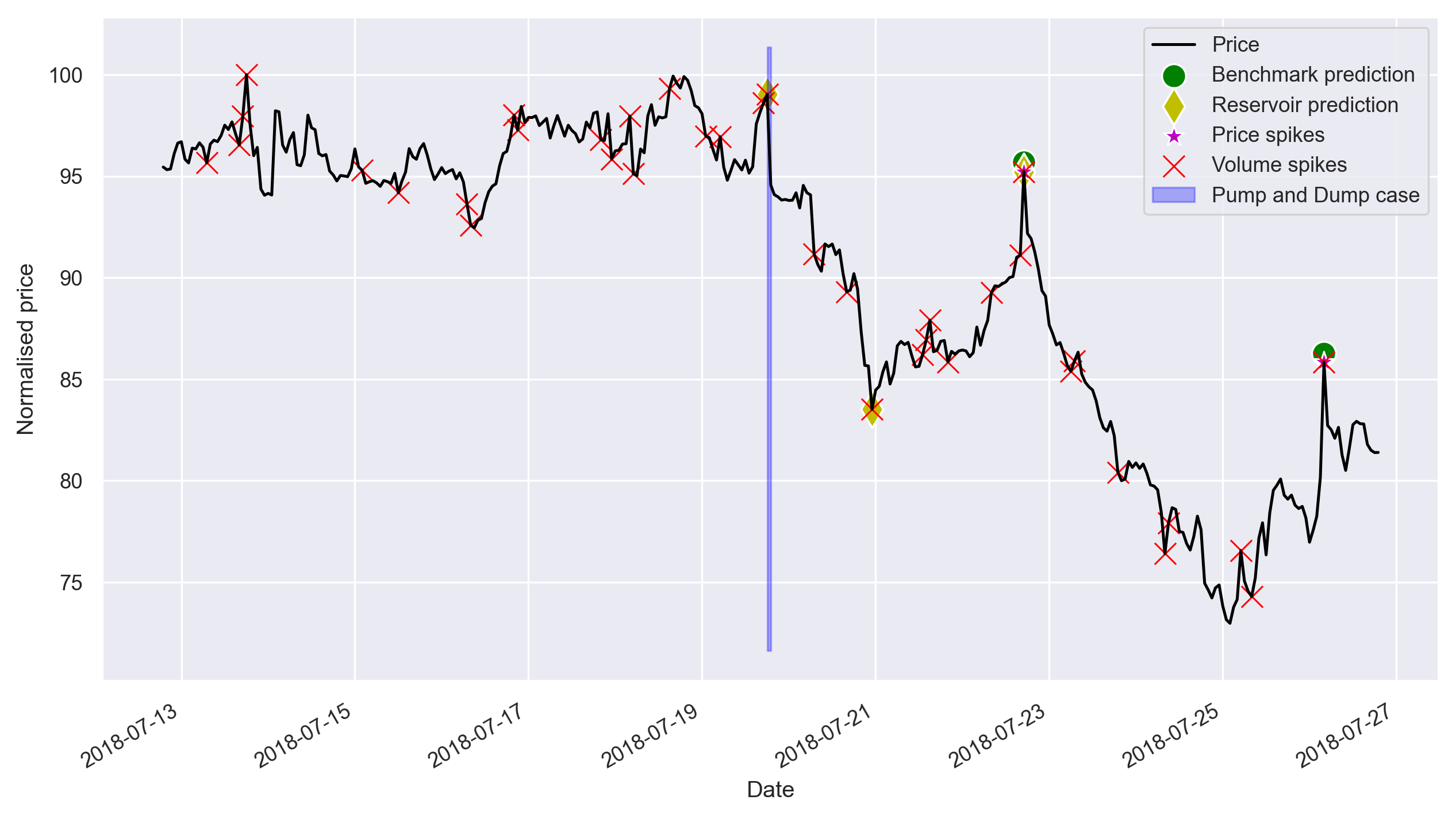}
\end{figure}

In the following, we present our results using the complete data set described in Section \ref{sec: data analysis}. 
In order to compare our results with the ones from \cite{LaMorgia2021}, we first consider 3 days interval for each PD attempt. 
As Figure \ref{fig: random reservoir pr 1.5} clearly shows both the isolation forest predictions and the robust covariance predictions using random signatures as a feature set outperform the benchmark model for almost every given recall.
\begin{figure}[H]
  \caption{Results for random signatures for both the isolation forest and robust covariance methods. The figures show precision-recall plots for the random signature predictions. The vertical lines show the respective recall values, where the maximum of the $F_1$ scores was achieved.}
  \includegraphics[width=0.9\textwidth]{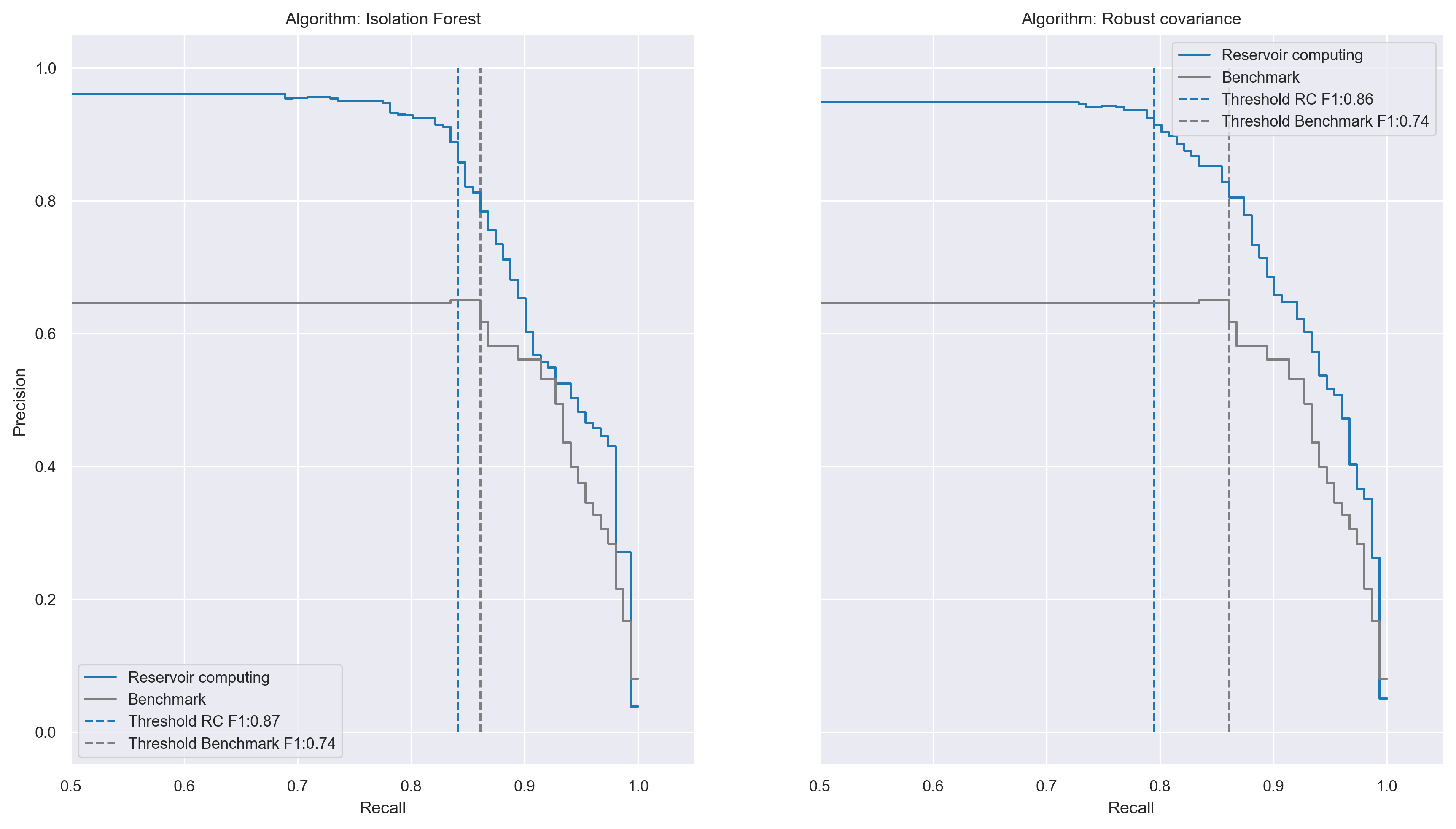}
  \label{fig: random reservoir pr 1.5}
\end{figure}

Figure \ref{fig: exact reservoir pr 1.5} shows, that the isolation forest predictions using exact truncated signatures shows similar results to the ones obtained using random signatures. On the other hand, the robust covariance predictions using exact truncated signatures under-performs even compared to the benchmark. The problem with the robust covariance method using exact signatures features is a consequence of the matrix bad condition number which shows warnings during the training step.
It is interesting to note that we did not observe the same result for random signatures and this can be explained by the fact that randomness can act as a method for implicit regularisation. 
Two possible references for this behavior are \cite{implicitreg}, in the context of shallow neural networks, and \cite{jacot2020implicit}, for random features extracted by a Gaussian process used for kernel regression.
Explicit regularisation of the exact reservoirs might help with the bad condition number but experiments of this kind have not been performed by the authors.
\begin{figure}[h]
  \caption{Results for exact truncated signatures for both the isolation forest and robust covariance methods. Compared to figure \ref{fig: random reservoir pr 1.5}, the robust covariance method algorithm significantly worse while the isolation forest algorithm shows similar results.}
  \label{fig: exact reservoir pr 1.5}
  \includegraphics[width=0.9\textwidth]{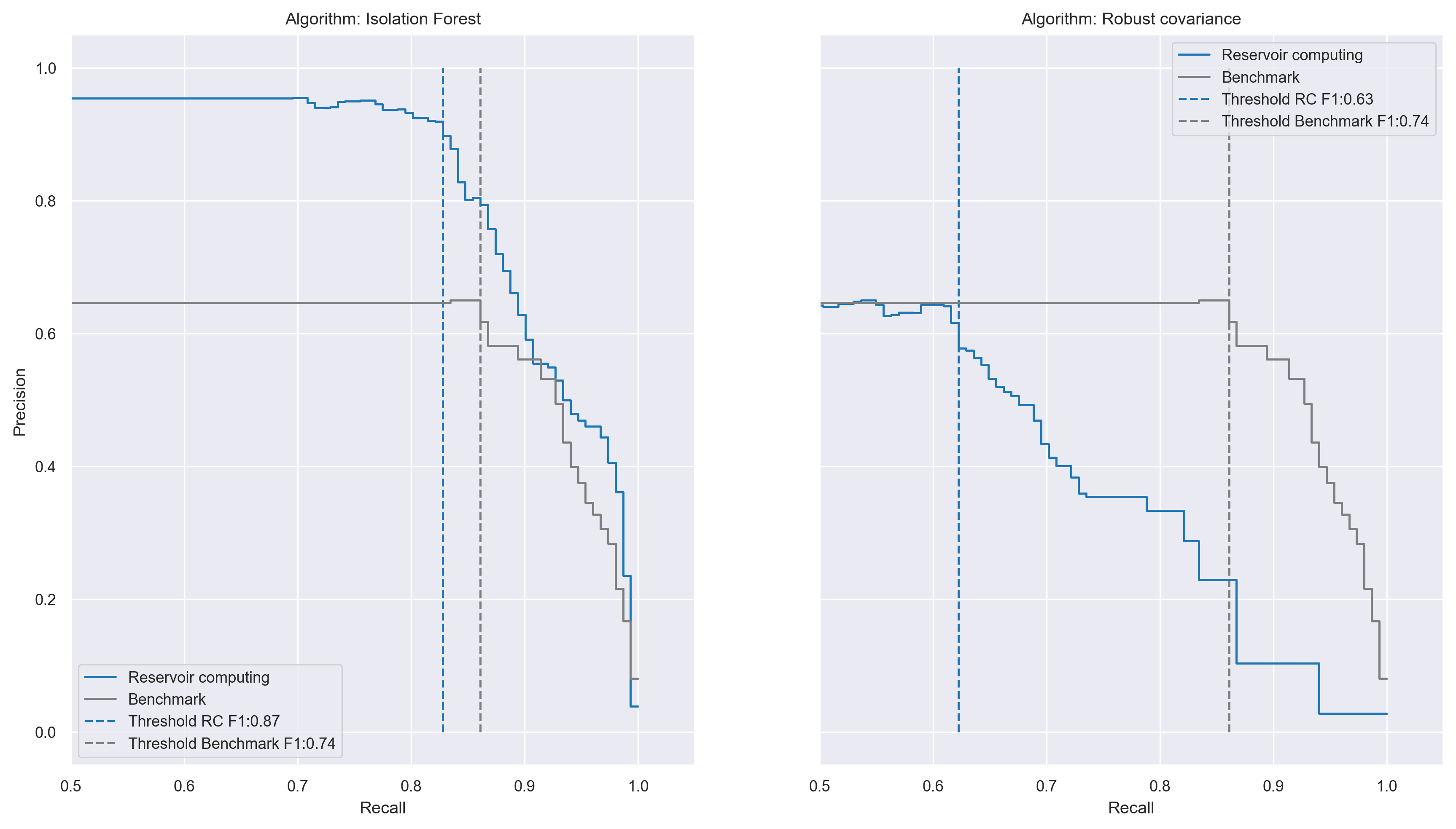}
\end{figure}

The fact that results for isolation forest show comparable results suggests that random signatures and truncated exact signatures contain a  similar amount of information. This is unsurprising, as the randomized signatures are an approximation of the exact signatures in a lower dimension and should therefore approximate the information contained in the exact signature.

Table \ref{tab: full results pd } shows the full results obtained in our experiments. As expected, for every classifier, the results get worse the detection duration increases. The main reason for this is that the chance of us missing PD attempts in our data set increases and also there are more unusual market behaviors which makes the classification harder in general. 

For isolation forests, both the randomized signatures and the truncated signatures show similar results for each time window. As previously discussed, the exact signatures show poor performance for the robust covariance algorithm.

Excluding the results for using exact truncated signatures with the robust covariance algorithm, our results outperform the results in \cite{Kamps2018} for all time windows. Due to the complexity of their method, we did not fully replicate the results described in \cite{LaMorgia2021} and hence we are only able to compare the results for the $\pm$ 1.5 days time window. Even their results are significantly better in terms of $F_1$ score, our best classifier comes relatively close. Here we would like to emphasize that the classifier used in \cite{LaMorgia2021} is based on supervised learning, while our classifier and the one in \cite{Kamps2018} are based on unsupervised learning.

Compared to supervised learning, our methodology has the big advantage, that there is no need for labeled training data and can hence be used without the need for manually collecting pump and dump cases. Additionally, given the generic nature of our method, it can potentially also be employed for the detection of other kinds of market anomalies.

\begin{table}[H]
\small
\centering
{
\caption{Anomaly detection performance}
\begin{tabular}{cccccc}
\toprule
    Classifier & Days & Precision & Recall & $F_1$\\
 \hline
Kamps et al. & 3&65\%&86\%&74\%\\
& 6&48\%&86\%&62\%&\\
& 14&29\%&81\%&43\%\\
    \hline
    La Morgia et al.\footnote{In their methodology, the authors use cross validation in order to compute the out of sample performance. As our methodology does not have use a training / test set split, results of might vary slightly even though the same overall data set is used. Therefore, the results here are only of an indicative nature.} & 3&98\%& 91\%& 95\%\\
& 6&-&-&-\\
& 14&-&-&-\\
        \hline
Randomized signature (Isolation forest) & 3&93\%&83\%&88\%\\
& 6&82\%&83\%&82\%\\
& 14&75\%&75\%&75\%\\
    \hline
Randomized signature (Robust covariance) & 3&80\%&94\%&86\%\\
& 6&80\%&84\%&82\%\\
& 14&75\%&75\%&75\%\\
    \hline
Exact truncated signature (Isolation forest)& 3&83\%&92\%&87\%\\
& 6&81\%&82\%&81\%\\
& 14&71\%&74\%&72\%\\
    \hline
Exact truncated signature (Robust covariance)& 
3&61\%&61\%&60\%\\
& 6&61\%&47\%&54\%\\
& 14&58\%&32\%&44\%\\

\bottomrule
\end{tabular}
\newline
\label{tab: full results pd }

}
     {\raggedright Table \ref{tab: full results pd } shows the full results for our classifiers together with the results from \cite{Kamps2018} and \cite{LaMorgia2021}. The results from \cite{LaMorgia2021} were not replicated by the authors and are taken directly from their paper as comparison. Here for each methodology and detection duration, the maximum $F_1$ score together with the precision and recall score at which it was obtained are displayed.\par}
\end{table}

\section{Conclusion}\label{sec:conclusion}
In this article, we explore the usage of randomized and truncated signatures as a feature set for anomaly detection purposes. 
As an empirical application, we first show how to detect real and fake trajectories of stock prices which are indistinguishable by visual inspection. 
Second, we investigate pump and dump schemes in the cryptocurrency ecosystem. 
Our analysis shows that both randomized and truncated signatures yield promising results as non-linear inputs for well established anomaly detection algorithms.
In particular, our work reveals that the use of signatures can significantly improve anomaly detection algorithms purely based on price and volume spikes. 
Furthermore, given the remarkable level of precision and recall achieved by our unsupervised learning algorithm, which are comparable to the results obtained from supervised learning, we are able to show that all necessary information to identify the pump and dump operations in the cryptocurrency market is already inherited in the historical time series of the related trades. 
More specifically, historical values of trade price, volume, date\&time, and side information are the only inputs in our machine learning algorithm.

Our findings also have important implications for the market efficiency. 
Clearly, cryptocurrency  markets attract increasing attention from the investment community due to the recent advances in their underlying technology. 
Cryptocurrencies now constitute an important asset class both for researchers and traders. 
In particular, the introduction of related derivative products will complete the market; enlargement by liquidity pools and non-fungible tokens will help to increase market efficiency. 
Since our work shows that unsupervised learning algorithms, which only use publicly available information to detect anomalies, achieve a similar performance as the supervised learning algorithms, market efficiency is supported. 
In fact, this carries an important message both for market designers and regulators since they can use our work to design or improve the market conditions so that certain types of anomalies do not appear. 
Furthermore, we also show that signatures are very successful in extracting valuable features from the input data hence this can also be used by traders to develop successful investment strategies.

Last but not least, the small difference between the results obtained by truncated signatures and randomized signatures empirically confirms the theory developed in Section \ref{sec:methods}. 
As the tasks explored in this article are relatively low dimensional, it is still feasible to compute exact signatures up to a relatively high degree. 
As the number of coefficients for the exact signatures grows exponentially in the degree (and polynomially in the path dimension), higher dimensional problems can make the computation of exact signatures infeasible. 
For these sorts of problems, for example in \cite{dyadic}, dyadic signatures have been used, which approximates higher order signatures using lower order signatures on path segments. 
In this article, we show that randomized signatures can be an alternative method to approximate the exact signatures. 
This, in turn, also serves as a robustness check for our newly developed methodology of randomized signatures. 

As future work, it would be interesting to preprocess the data to reflect the prior knowledge that we have about PD attempts by prescribing a certain kind of pattern as the one described in Section \ref{sec: pd process}. 
Hence, our algorithm can be improved by processing the input data to reflect the prior knowledge. 
Other possible future directions for this research are to investigate whether the methods can be improved by considering multiple currency pairs as inputs, which would also take correlations into account, or to modify outlier detection algorithms with new and more sophisticated variants.

\newpage
\bibliographystyle{elsarticle-harv}\biboptions{authoryear}
\bibliography{rsignature.bib}

\listoffigures

\end{document}